\documentclass[conference]{IEEEtran}
\makeatletter
\def\ps@headings{%
\def\@oddhead{\mbox{}\scriptsize\rightmark \hfil \thepage}%
\def\@evenhead{\scriptsize\thepage \hfil \leftmark\mbox{}}%
\def\@oddfoot{}%
\def\@evenfoot{}}
\makeatother
\renewcommand{\baselinestretch}{0.90}
\usepackage{graphicx}
\usepackage{balance}
\usepackage{comment}
\usepackage{amsmath}
\usepackage{amssymb}
\usepackage{hyperref}

\newcommand{\remark}[1]{{\emph{Remark:}} #1}

\newcommand{\BEQA}{\begin{eqnarray}}
\newcommand{\EEQA}{\end{eqnarray}}
\newcommand{\define}{\stackrel{\triangle}{=}}

\newcommand{\gap}{\vspace{2mm}}
\newcommand{\Gt}{\tilde{G}}
\newcommand{\Rt}{\tilde{R}}
\newcommand{\Vt}{\tilde{V}}

\newcommand{\F}{\mathcal{F}}

\newcommand{\U}{\mathcal{U}}
\renewcommand{\P}{\mathcal{P}}

\newcommand{\mbar}{\overline{m}}
\newcommand{\Btilt}{\tilde{B}}
\newtheorem{theorem}{Theorem}
\newtheorem{proposition}{Proposition}
\newtheorem{lemma}{Lemma}

\newtheorem{definition}{Definition}

\newtheorem{corollary}{Corollary}

\title{QoS Constrained Optimal Sink and Relay Placement in Planned Wireless Sensor Networks}

\author{\IEEEauthorblockN{Abhijit Bhattacharya, Akhila Rao, 
    Naveen K. P., Nishanth P. P., \\
S.V.R. Anand, and Anurag Kumar} \IEEEauthorblockA{Dept. of
    Electrical Communication Engineering,
    Indian Institute of Science (IISc), Bangalore 560012, India.\\
    Email: \{abhijit, naveenkp, anand,
    anurag\}@ece.iisc.ernet.in, \{akhila.suresh.rao, nishanth.pp93\}@gmail.com}}

\begin{document}

\maketitle
\vspace{-8mm}
\begin{abstract}
\label{abstract}
We are given a set of sensors at given locations, a set of potential locations for placing base stations (BSs, or sinks), and another set of potential locations for placing wireless relay nodes. There is a cost for placing a BS and a cost for placing a relay. The problem we consider is to select a set of BS locations, a set of relay locations, and an association of sensor nodes with the selected BS locations, so that number of hops in the path from each sensor to its BS is bounded by $h_{\max}$, and among all such feasible networks, the cost of the selected network is the minimum. The hop count bound suffices to ensure a certain probability of the data being delivered to the BS within a given maximum delay under a light traffic model. We observe that the problem is NP-Hard, and is hard to even approximate within a constant factor. For this problem, we propose a polynomial time approximation algorithm (SmartSelect) based on a relay placement algorithm proposed in our earlier work, along with a modification of the greedy algorithm for weighted set cover. We have analyzed the worst case approximation guarantee for this algorithm. We have also proposed a polynomial time heuristic to improve upon the solution provided by SmartSelect. Our numerical results demonstrate that the algorithms provide good quality solutions using very little computation time in various randomly generated network scenarios.
\end{abstract}

\begin{keywords}
  Wireless sensor network design; Multiple sink and relay placement; QoS-aware network design
\end{keywords}

\section{Introduction}
\label{sec:intro}

\subsection{Motivation and Problem Definition} 
Industrial monitoring and control applications typically have a large number of sensors distributed over hundreds of meters from the control center. Traditionally, the sensor readings are communicated to the control center via wired networks that are difficult to install and maintain. Therefore, recently there has been increasing interest in replacing these wireline networks with wireless packet networks (\cite{honey,isa,ge}). 

Owing to the small communication range of the sensing nodes (typically a few tens of meters depending on the RF propagation characteristics of the environment), usually multi-hopping is needed to communicate to the control center. The practical problem that we consider in this paper is the following:

\begin{enumerate}
\item There are already deployed, static sensors from which measurements, encapsulated into packets, need to be collected. We also refer to the sensors as \emph{sources.}
\item Additional relays and base stations (BS) need to be placed in the region in order to provide multi-hop paths from each of the sources to at least one BS. The sources can also act as relays for the packets from other sources. The network so obtained needs to provide certain quality-of-service (QoS) to the packets flowing over it, in terms of, e.g., delivery probability, or packet delay. 

In most practical applications, due to the presence of obstacles to radio propagation, or due to taboo regions, we cannot place relays and sinks anywhere in the region, but only at certain designated locations. This leads to the problem of \emph{constrained node placement} in which the nodes are constrained to be placed at certain \emph{potential locations}. Further, only certain links are permitted\footnote{this could be because some links could be too long, leading to high bit error rate and hence large packet delay, or due to an obstacle, e.g., a firewall}. See Figure~\ref{fig:constrained} for a depiction of the problem.

\begin{figure}[t]
\begin{center}
\includegraphics[scale=0.4]{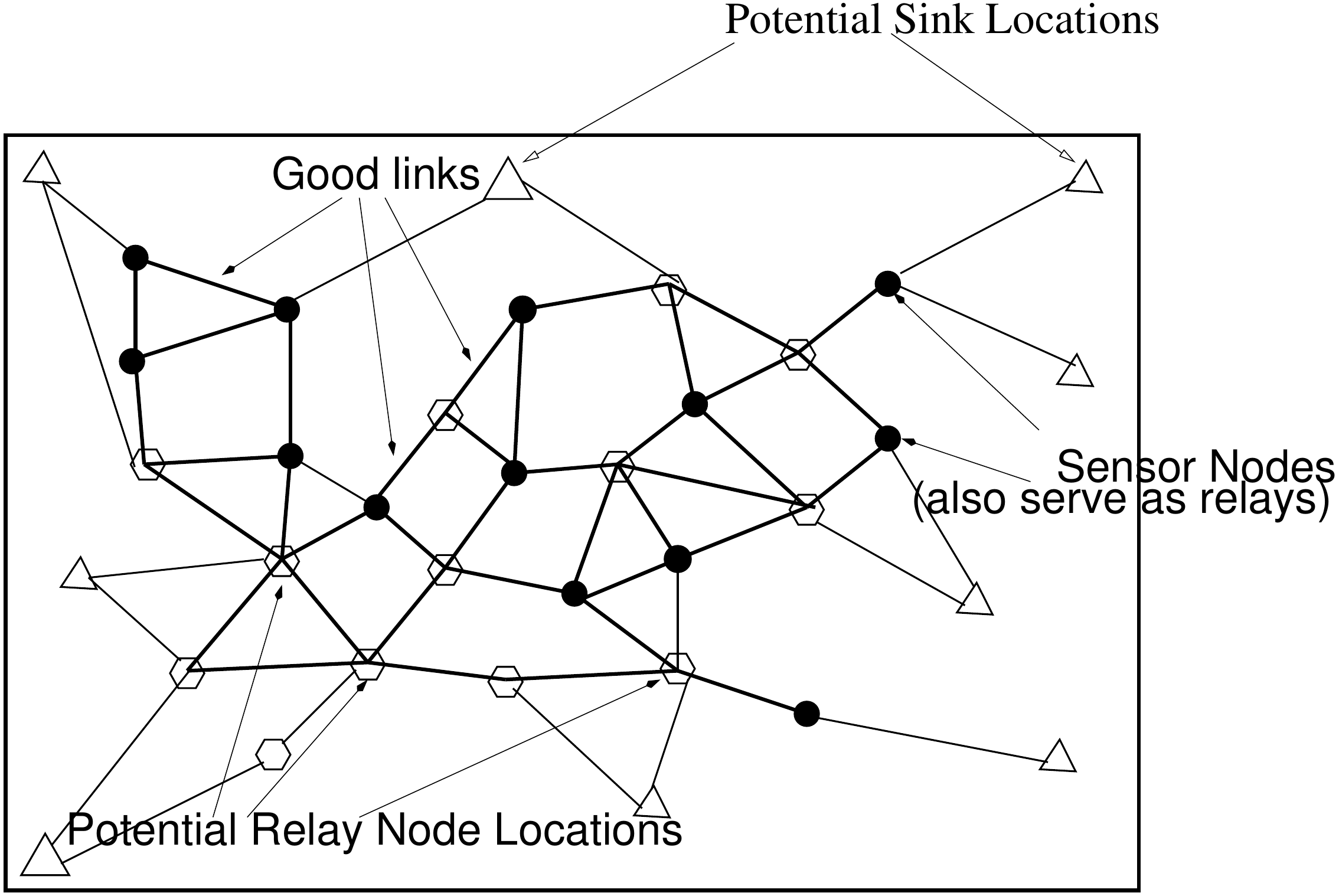}
\end{center}
\caption{The constrained sink and relay placement problem; circles indicate sources, hexagons indicate \emph{potential} relay locations, and triangles indicate \emph{potential} sink locations. The edges denote the useful links between the nodes.}
\label{fig:constrained}
\end{figure} 

\item We assume that there is a cost associated with each sink, and each relay. The objective of the design is to \emph{place a minimum cost selection of sinks and relays} (at the potential locations) while achieving a network that meets the following QoS objectives:
\begin{enumerate}
\item There is a path from each source node to at least one BS.
\item The maximum delay on any path is bounded by a given value $d_{\max}$, and the packet delivery probability (the probability of delivering a packet within the delay bound) on any path is $\geq p_{\mathrm{del}}$.
\end{enumerate}
\end{enumerate}

Note that in wireless networks, the actual link qualities are apriori unknown. Stochastic models for radio propagation are approximate, and therefore, designs based only on such models cannot be guaranteed to work when deployed on field. The ``usable'' network links (e.g., links that meet a certain target packet error rate requirement) can only be ascertained by field measurements. This motivates iterative algorithms that start with a model based design, make partial deployments and explore the field by making measurements, thereby updating the knowledge about on-field link qualities, and repeat the procedure until a design is obtained that satisfies the QoS requirements on field (see, for example, \cite{smartconnect-paper}). Such iterative algorithms require to invoke at each iteration, a sub-algorithm for network design, which takes as input, a network graph with \emph{given} link qualities, and outputs a subgraph that meets the QoS requirements assuming the given link qualities to be true. Since this network design module has to be invoked at each iteration, and there could be several such iterations before the final design is obtained, it is desirable that this module has \emph{low time complexity}. We therefore, seek \emph{fast} algorithms that provide optimal or close to optimal solutions. In our earlier work \cite{iwqos,fullpaper,smartconnect-paper}, we addressed this problem for the special case where there is a single BS at a designated location. In this paper, we study the more general problem where there is the possibility of deploying multiple sinks. This general problem is of interest for scalability of the network design, i.e., when we want to deploy a network over an area so large that a single sink based solution fails to meet the desired QoS requirements (for example, since in low power wireless networks the usable link lengths are limited, with just one sink the number of hops from some of the sources to the sink can become so large as to make the packet delivery probability unacceptably small).
 
\subsection{Organization of the Paper}
The rest of the paper is organized as follows: in Section~\ref{sec:formulation}, we describe the problem formulation, discuss the complexity of the problem, and provide a brief summary of related literature. Section~\ref{sec:smartselect} presents a polynomial time approximation algorithm (SmartSelect) for the problem, and provides a worst case analysis of the algorithm. We also provide a polynomial time heuristic to improve upon the solution obtained using SmartSelect. In Section~\ref{sec:results}, we provide extensive numerical results for our algorithms applied to a set of random scenarios. Finally, we conclude the paper in Section~\ref{sec:conclude}.

\section{The Network Design Problem}
\label{sec:formulation}

\subsection{The Lone Packet Model}
In this paper, we address the problem for the case where the traffic from the source nodes is \emph{very light}. Formally, we define ``light traffic'' as follows: for any set of links (transmitter-receiver pairs) whose transmissions interfere with one another, only one of the links is active (i.e., carries a packet) in the network at any point of time. We call this the ``lone packet traffic model,'' which is realistic for many applications, including the so called \emph{condition monitoring/industrial telemetry} applications (\cite{telemetry1,telemetry2}), where the time between successive measurements being taken is sufficiently long so that the measurements can be staggered so as not to occupy the medium at the same time. The motivation behind the lone-packet model comes from the following important result (formally proved in \cite{fullpaper}): for a design (network) to satisfy the QoS objectives for a given positive arrival rate (continuous traffic), it is \emph{necessary} that the network satisfies the QoS objectives under zero/light traffic load, i.e., the lone packet model. Hence, we cannot hope to solve the general problem of QoS aware network design for positive traffic arrival rates unless we have a reasonably good solution to the more basic problem of lone packet based network design. Also note that such designs based on lone-packet model can be used as a starting point for network design with more general arrival processes. Packet level simulation results reported in \cite{fullpaper} suggest that such lone-packet model based designs suffice for some small positive traffic arrival rates (as are typical of non-critical monitoring and control applications). For a more detailed discussion on the applicability and justification of the lone-packet model, see \cite{fullpaper}. 
 
\subsection{The Network Design Setting}
\label{subsec:network-setting}
In this subsection, we discuss how we can map the packet level QoS objectives into graph level objectives under the lone packet model. Given a set of source nodes or required vertices $Q$, a set of potential relay locations $R$, each with cost $c_r$, and a set of potential sink locations $B$, each with cost $c_s$, we consider a graph $G= (V, E)$ on $V= Q\cup R\cup B$ with $E$ consisting of all \emph{feasible} edges. 

Throughout this work, we assume that all nodes operate at the same fixed power level. We can then define the set of \emph{feasible} edges $E$, either by imposing a bound on the packet error rate (PER) of each link, or alternately, by constraining the maximum allowed link length (which, in turn, affects the link PER). Having thus characterized the link quality of each feasible link in the graph $G$, it can be shown by an elementary analysis that the QoS objectives ($d_{\max}$ and $p_{\mathrm{del}}$) can be met by imposing a hop count bound of $h_{\max}$ between each source node and the sink.  Details of this analysis are provided in \cite{fullpaper}, where we have considered the practical situation of slowly fading links, and packet losses due to random channel errors. Thus, there is a random delay at each hop due to packet retransmissions, and packets could be dropped if a retransmission limit is reached. Note that as a consequence of the lone packet assumption, the delay along a path is additive, i.e., it is simply the sum of the delays on each hop along the path. The analysis uses statistical models (which can be obtained from field measurements) for link errors and random fading, as well as parameters of the wireless physical layer and the medium access control used by the system (e.g., the back-off parameters of the IEEE~802.15.4 CSMA/CA \cite{IEEE}). 

\subsection{Problem Formulation}
\label{subsec:formulation}

Given the graph $G= (V, E)$ on $V= Q\cup R\cup B$ with $E$ consisting of all \emph{feasible} edges (as explained in Section~\ref{subsec:network-setting}), costs $c_s$ and $c_r$ of each sink and relay respectively, and a hop constraint $h_{\max}$, \emph{the problem is to extract from this graph, a minimum cost (cost of selecting sinks and relays) subgraph spanning $Q$, such that each source has a path to at least one sink with hop count $\leq h_{max}$.} We call this the \emph{MultiSink Steiner Network-Minimum Cost-Hop Constraint} (MSSN-MC-HC) problem.\footnote{In this formulation, we have not taken into account the energy expenditure at a node due to transmission and reception since, in a light traffic setting, the fraction of time a node is transmitting or receiving a packet is small compared to the idle time of the node.}

\subsection{Complexity of the Problem}

\begin{proposition}
\begin{enumerate}
\item \textbf{Complexity: }The MSSN-MC-HC problem is NP-Hard. 
\item \textbf{Inapproximability: }The MSSN-MC-HC problem is not constant factor approximable. In particular, it cannot be approximated to within a factor better than $O(\log(m))$, where $m$ is the number of sensors.
\end{enumerate}
\end{proposition}

\begin{proof}
\begin{enumerate}

\item The proof is via a restriction argument \cite[p. 63, Section 3.2.1]{Garey}, and a reduction from set cover. Consider the subclass of MSSN-MC-HC problems where $c_r = 0$. The problem now is to select a minimum number of sinks from the set $B$ such that each source is connected to at least one sink with hop count $\leq h_{\max}$. We call this the \emph{sink placement} problem.  

Consider now an instance of the set cover problem where we are given a set of elements, say $Q=\{q_1,\ldots, q_m\}$, and a set $\tilde{B}$ of subsets of $Q$, say $\Btilt=\{\tilde{b}_1,\ldots, \tilde{b}_n\}$. The problem is to select a minimum number of sets from $\Btilt$ such that their union covers $Q$. We can \emph{reduce this set cover problem to an instance of the sink placement problem as follows}: associate a sensor node with each element of $Q$, and a potential sink node, say $b_k$, with each element $\tilde{b}_k$ of $\Btilt$, $k=1,\ldots,n$. If $\tilde{b}_k = \{q_{k1},\ldots,q_{k,m_k}\}$, then connect the sensors $q_{k1},\ldots,q_{k,m_k}$ to the sink $b_k$, using at most $h_{\max}$ hops. Note that since $c_r=0$, we are free to use as many relays as we wish. Thus, the set cover problem reduces to that of selecting a minimum number of sinks from the set $B\define\{b_1,\ldots,b_n\}$, such that each sensor in $Q$ is covered by (i.e., connected via $\leq h_{\max}$ hops to) at least one sink, which is precisely an instance of the sink placement problem. 

Now, since we know that the set cover problem is NP-Hard, and since any instance of the set cover problem can be reduced to an instance of the sink placement problem as above, it follows that \emph{the sink placement problem is NP-Hard}. Hence, by the restriction argument, it follows that MSSN-MC-HC problem is NP-Hard. 

\item Moreover, note that the above reduction from set cover problem to sink placement problem is \emph{approximation preserving}, i.e., if we have an $\alpha$-approximation algorithm for the sink placement problem, it will give an $\alpha$-approximation for the set cover problem. But it is also known that the best possible approximation factor for the set cover problem is $O(\log(m))$, where $|Q|=m$. Hence, it follows that the best possible approximation factor for the sink placement problem is $O(\log(m))$. Since the class of sink placement problems is a strict subset of the class of MSSN-MC-HC problems, it follows that the MSSN-MC-HC problem cannot be approximated to within a factor of $O(\log(m))$.
\end{enumerate}
\end{proof}

\begin{proposition}
If $\frac{c_s}{c_r}\geq \mbar(\mbar+1)(h_{\max}-1)$, for some $\mbar\in\mathbb{N}$, $m > \mbar \geq 1$, then the worst case approximation guarantee of any algorithm for the MSSN-MC-HC problem is upper bounded by $m\left(1+\frac{1}{\mbar(\mbar + 1)}\right)$, where $|Q|=m$. 
\end{proposition}
\begin{proof}
Observe that the cost of the optimal solution to the MSSN-MC-HC problem is lower bounded by $c_s$ (since at least one sink is required), and the cost of the outcome of any algorithm is upper bounded by $mc_s\: +\: m(h_{\max}-1)c_r$ (corresponding to the case where each source is connected to a different sink using relay-disjoint paths, each path having $h_{max}-1$ relays). Hence, the worst case approximation guarantee of any algorithm is upper bounded by

\begin{align}
\frac{mc_s\: +\: m(h_{\max}-1)c_r}{c_s}&= m\:+\:m(h_{\max}-1)\frac{c_r}{c_s}\nonumber\\
&\leq m\:+\:\frac{m}{\mbar(\mbar + 1)}\nonumber\\
&= m\left(1+\frac{1}{\mbar(\mbar + 1)}\right)\nonumber
\end{align}
\end{proof}

\subsection{Related Work}
Several variations of the optimal node placement problem have been studied in the literature. \cite{Lin,cheng,Lloyd} studied variations of the problem of \emph{unconstrained relay placement for connectivity and/or survivability}, where a minimum number of relays have to be placed in a two-dimensional region (they can be placed anywhere; no \emph{potential} locations are given) to obtain a tree spanning a given set of \emph{required} vertices (sources and BS). No QoS constraint was imposed in their formulations. They showed the problems to be NP-Hard, and proposed approximation algorithms. Bredin et al.\ \cite{bredin} studied the problem of optimal relay placement \emph{(unconstrained)} for $k-$connectivity. They proposed an $O(1)$ approximation algorithm for the problem with any \emph{fixed} $k\geq 1$. However, they also did not impose any QoS constraint. 

\cite{Misra,yang} studied variations of the \emph{constrained relay placement problem} for connectivity and survivability. Although their formulations take into account an edge length bound, which can model the link quality, the formulation does not involve a path constraint such as the hop count along the path; hence, there is \emph{no constraint on the end-to-end QoS.} They showed that the problems are NP-Hard, and proposed $O(1)$ (respectively, $O(\ln n)$) approximation algorithms. 

In \cite{iwqos,fullpaper}, we have studied the problem of \emph{constrained} relay placement \emph{with end-to-end QoS objective}, but with a single Base Station at a given location. Our formulation converts the QoS objective into a hop constraint on each source-sink path under the lone-packet model. We showed that the problem is NP-Hard, and proposed a polynomial time approximation algorithm (SPTiRP). In \cite{smartconnect-paper}, we used the SPTiRP algorithm to build a  \emph{field-interactive}, iterative network design tool (SmartConnect) for designing QoS-aware networks. In our current work, we aim at extending this formulation to incorporate the possibility of deploying multiple sinks. 
  
Recently, Sitanayah et al.\ \cite{sitanayah} have aimed to address this problem of multiple sink and relay placement with hop constraint. However, they have proposed only a local search heuristic (GRASP-MSRP) of \emph{exponential} time-complexity for the problem. No theoretical study of either the problem, or the proposed algorithm was provided. Although their local search heuristic works well in practice, \emph{the complexity of the algorithm prohibits its use in an iterative network design process} such as SmartConnect\cite{smartconnect-paper}. Hence, we seek, instead, \emph{fast} heuristics that perform reasonably close to optimal. \emph{Note that our heuristics can also be used as an initialization in the GRASP-MSRP algorithm \cite{sitanayah} to improve its running time, thus yielding the best of both time and performance}.

\section{MSSN-MC-HC: A Heuristic and its Analysis}
\label{sec:smartselect}
In this section, we present a polynomial time approximation algorithm for the MSSN-MC-HC problem. The algorithm proceeds by reducing the problem to a modified version of the weighted set cover problem, and the greedy algorithm for weighted set cover is used to obtain a solution. Note that since the greedy algorithm for weighted set cover is polynomial time \cite{vazirani}, the proposed algorithm is polynomial time. 

\subsection{SmartSelect: A Greedy Algorithm for Sink and Relay Selection}

\begin{enumerate}
\item \label{step:zero-relay}\textbf{The single sink, zero relay case: }Consider the restriction of the graph $G$ to only the sources, $Q$, and the potential sinks, $B$. For each sink $b\in B$, find the shortest path tree rooted at $b$ spanning the sources in $Q$. If there exists a sink $b_0$ such that the SPT rooted at $b_0$ satisfies the hop constraint for each source in $Q$, then we are done; the optimal solution requires a single sink, and no relays. Otherwise, go to the next step.
\item \label{step:spt}\textbf{Checking feasibility: }On graph $G$ (i.e., now including all sources and relays), obtain a shortest path tree rooted at each potential sink location (i.e., we have as many shortest path trees as there are potential sinks). If there exists a source such that its shortest paths to all the sinks have lengths exceeding $h_{\max}$, declare the problem infeasible. Else, go to the next step.
\item For each sink $b_i\in B$, $i=1,\ldots,|B|$, identify the set of sources, say $Q_i$, whose shortest paths to $b_i$ have lengths $\leq h_{\max}$. The set $Q_i\subseteq Q$, is said to be \emph{covered} by $b_i$. Note that, having ensured feasibility in Step 2, $\cup_{i=1}^{|B|}Q_i=Q$. Also identify the set of relays, $R_i\subseteq R$, whose shortest paths to $b_i$ have lengths $\leq h_{\max}-1$ (this helps to reduce the complexity of Step 4; indeed $R_i$ is the set of relays that may ever be used to connect the sources in $Q_i$ to $b_i$). 
\item Set $j\leftarrow 0$. The iterations will be indexed by $j$. Set $Q_i^{(0)} = Q_i$, $B^{(0)}=B$. $B^{(j)}$ denotes the set of sinks not yet picked at the start of iteration $j$, $j\geq 0$, and $Q_i^{(j)}$ denotes the set of \emph{uncovered} sources that are associated with a BS $b_i\in B^{(j)}$ at the start of iteration $j$.

\gap
\noindent
\textbf{The greedy iterative algorithm:}
\item\label{step:sptirp} For each $i$ such that $b_i\in B^{(j)}$, let $G^{(j)}_i$ be the restriction of $G$ to $Q_i^{(j)}\cup R_i\cup\{b_i\}$. Run an algorithm (e.g., the SPTiRP algorithm \cite{fullpaper}) on $G^{(j)}_i$ to obtain a \emph{near-optimal} subset of relays, $\hat{R}^{(j)}_i\subseteq R_i$, that connect the sources in $Q^{(j)}_i$ to $b_i$ with $\leq h_{\max}$ hops. 
\item For each $i$ such that $b_i\in B^{(j)}$, consider the restriction of $G^{(j)}_i$ to $Q^{(j)}_i\cup \hat{R}^{(j)}_i\cup b_i$. Denote this graph by $\tilde{G}^{(j)}_i$.
\item For each $i$ such that $b_i\in B^{(j)}$, define the cost of $\tilde{G}^{(j)}_i$ as 
\begin{equation}
C_i^{(j)} = \frac{c_s+c_r\times |\hat{R}^{(j)}_i|}{|Q^{(j)}_i|}
\end{equation}
i.e., the cost of a subgraph is computed as the total cost per source.
\item \textbf{The greedy selection: }Pick the subgraph with the least cost among the subgraphs not yet picked, i.e., pick $\tilde{G}^{(j)}=\arg\min_{\tilde{G}^{(j)}_i}C^{(j)}_i$. Break ties by picking the subgraph that covers more sources. Let $\tilde{b}^{(j)}$ be the sink associated with $\tilde{G}^{(j)}$. 
\item Let $Q^{(j)}=Q\cap \tilde{G}^{(j)}$, $R^{(j)}=R\cap \tilde{G}^{(j)}$. If $\cup_{k=1}^{j}Q^{(k)}=Q$, \textbf{STOP}, i.e., stop when all the sources have been \emph{covered}. Else, go to next step.
\item \textbf{Update step: }Set $B^{(j+1)}=B^{(j)}\backslash \tilde{b}^{(j)}$. For each $i$ such that $b_i\in B^{(j+1)}$, set $Q_i^{(j+1)}=Q_i^{(j)}\backslash (\tilde{G}^{(j)}\cap Q_i^{(j)})$. Moreover, set the cost of each relay in $R^{(j)}$ to zero for all future iterations. This is done so that sources and relays that are shared by covers do not get counted more than once.
\item Set $j\leftarrow j+1$, and go to Step~\ref{step:sptirp}.  
\end{enumerate}

\subsection{Analysis of SmartSelect}
\subsubsection{Some observations}

\begin{enumerate}
\item If the optimal solution uses a single sink, and no relays, SmartSelect \emph{achieves the optimal solution} (follows from Step~\ref{step:zero-relay}).
\item For the sink placement problem (i.e., $c_r=0$), the SmartSelect algorithm reduces exactly to the greedy algorithm for weighted set cover, and hence \emph{achieves the best possible worst case approximation guarantee ($O(\log(m))$) for the sink placement problem}, where $m$ is the number of sources, i.e., $|Q|=m$.
\item For the subclass of problems with $c_s = 0$, the worst case approximation guarantee is $m(h_{\max}-1)$ (if SPTiRP algorithm is used in Step~\ref{step:sptirp}).
\end{enumerate}

\subsubsection{Worst Case Approximation Guarantee}
We already know that when an optimal solution uses a single sink, and \emph{no relays}, the SmartSelect algorithm gives the optimal solution. We, therefore, focus on instances where any optimal solution uses at least one sink, and \emph{at least one relay}. We start with the following lemma.

\begin{lemma}
\label{lem:greedy-lemma}
Suppose $\frac{c_s}{c_r}\geq \mbar(\mbar+1)(h_{\max}-1)$, for some $\mbar\in\mathbb{N}$, $m > \mbar \geq 1$.\footnote{this assumption is not too restrictive because in practice, cost of a base station is much more than cost of a relay node, since setting up a base station typically requires infrastructure such as uninterrupted power supply from mains, ethernet or Wifi connectivity to a backhaul network etc.} Then the following holds: In any iteration of the SmartSelect algorithm, if more than $\mbar$ sources remain to be covered, the algorithm cannot favor a sink covering at most $\mbar$ remaining sources over a sink that covers \emph{all} the remaining sources, provided such a sink exists.
\end{lemma}

\begin{proof}
Suppose, upto the start of iteration $j$, $j\geq 1$, $m^{(j)}$ sources have been covered, and $m-m^{(j)}> \mbar$. Suppose there exists a sink (not yet picked) that covers all the remaining sources, and another sink (not yet picked) that covers $\mbar$ of the remaining sources. We index these two sinks by 1 and 2 respectively. Let $n_1^{(j)}$, and $n_2^{(j)}$ be the number of relays picked in Step 4 of the SmartSelect algorithm in iteration $j$ for sinks 1 and 2 respectively. Then, the algorithm favors sink 2 over sink 1 iff
\begin{align}
\frac{c_s+n_1^{(j)}c_r}{m-m^{(j)}} &> \frac{c_s+n_2^{(j)}c_r}{\mbar}\nonumber\\
\Leftrightarrow \frac{c_s}{c_r}&< \frac{\mbar n_1^{(j)}-(m-m^{(j)})n_2^{(j)}}{m-m^{(j)}-\mbar}\label{eqn:greedy-failure}
\end{align}
Now, we make the following observations:
\begin{enumerate}
\item $n_1^{(j)}\leq (m-m^{(j)})(h_{\max}-1)$, since each of the remaining sources connects to the sink 1 using at most $h_{\max}$ hops.
\item $\frac{\mbar (m-m^{(j)})}{m-m^{(j)}-\mbar}\:=\:\frac{\mbar}{1-\frac{\mbar}{m-m^{(j)}}}\:\leq\:\mbar(\mbar + 1)$, where the last inequality follows since $m-m^{(j)}\geq \mbar+1$. 
\end{enumerate} 

Then we have that

\begin{align}
\frac{c_s}{c_r}&\geq \mbar(\mbar+1)(h_{\max}-1)\nonumber\\
&\geq \frac{\mbar(m-m^{(j)})(h_{\max}-1)}{m-m^{(j)}-\mbar}\nonumber\\
&\geq \frac{\mbar n_1^{(j)}}{m-m^{(j)}-\mbar}\nonumber\\
&\geq \frac{\mbar n_1^{(j)}-(m-m^{(j)})n_2^{(j)}}{m-m^{(j)}-\mbar}\nonumber
\end{align}

Thus, condition~\eqref{eqn:greedy-failure} cannot hold, and hence, the algorithm cannot favor sink 2 over sink 1.

Proceeding similarly as above, and noting that $f(\mbar)\define\mbar(\mbar+1)$ is monotonically increasing in $\mbar$, it can be shown that the algorithm cannot favor a sink covering less than $\mbar$ sources over sink 1. This completes the proof of the lemma.
\end{proof}

Equipped with the above lemma, we derive bounds on the worst case approximation factor of the SmartSelect algorithm as follows.

\begin{theorem}
\label{thm:smartselect-approx_special}
Consider the subclass of MSSN-MC-HC problems where there exists a feasible solution that uses exactly one sink. Further assume that $\frac{c_s}{c_r}\geq \mbar(\mbar+1)(h_{\max}-1)$, for some $\mbar\in\mathbb{N}$, $m > \mbar \geq 1$. Then the following hold:
\begin{enumerate}
\item The number of sinks picked by the SmartSelect algorithm is at most $\lceil\frac{m}{\mbar + 1}\rceil$.
\item The worst case approximation guarantee of the SmartSelect algorithm is upper bounded by $\epsilon + \frac{m}{\mbar}$, where $\epsilon\in[0,1)$ is such that $\lceil\frac{m}{\mbar + 1}\rceil=\frac{m}{\mbar + 1}+\epsilon$.
\end{enumerate} 
\end{theorem}

\begin{proof}
\begin{enumerate}
\item Since there exists a feasible solution using exactly one sink, all the $m$ sources can be connected to that sink using $\leq h_{\max}$ hops. Hence, as long as this sink is not picked during the course of the SmartSelect algorithm, the hypothesis in Lemma~\ref{lem:greedy-lemma} continues to hold, i.e., there exists a sink that can cover \emph{all} the remaining sources. Moreover, as soon as this sink is picked (if at all), the algorithm terminates since all the sources get covered. Thus, by Lemma~\ref{lem:greedy-lemma}, as long as there are at least $\mbar+1$ sources remaining to be covered, the algorithm cannot pick any sink that covers at most $\mbar$ sources. Hence, the algorithm covers at least $\mbar + 1$ sources in each iteration until the number of remaining uncovered sources is $\leq \mbar$. Finally, when the number of uncovered sources is $\leq \mbar$, it follows (using Lemma~\ref{lem:greedy-lemma},and the monotonicity of $\mbar(\mbar + 1)$) that the algorithm covers these remaining sources using a single sink. Hence, the number of sinks picked by the algorithm is at most $\lceil\frac{m}{\mbar + 1}\rceil$.
\item Since an optimal solution uses at least one sink, and at least one relay \footnote{recall that for the ``single sink-zero relay'' case, SmartSelect gives optimal solution; hence we are not considering that case}, the optimum cost is at least $(c_s + c_r)$. Also, the number of relays used by the SmartSelect algorithm is trivially upper bounded by $m(h_{\max}-1)$. Thus, using the result from Part 1, the worst case approximation ratio is upper bounded as
\begin{align}
\text{Approx. ratio}&\leq\frac{\lceil\frac{m}{\mbar+1}\rceil c_s\:+\:m(h_{\max}-1)c_r}{c_s\:+\:c_r}\nonumber\\
&= \frac{(\frac{m}{\mbar + 1}+\epsilon) c_s/c_r\:+\:m(h_{\max}-1)}{1\:+\:c_s/c_r}\nonumber\\
&\leq \frac{(\frac{m}{\mbar + 1}+\epsilon)c_s/c_r\:+\:m(h_{\max}-1)}{c_s/c_r}\nonumber\\
&= (\frac{m}{\mbar + 1}+\epsilon)\:+\:\frac{m(h_{\max}-1)}{c_s/c_r}\nonumber\\
&\leq (\frac{m}{\mbar + 1}+\epsilon)\:+\:\frac{m}{\mbar(\mbar + 1)}\nonumber\\
&= \epsilon + \frac{m}{\mbar}\nonumber
\end{align}
where the last inequality follows since $\frac{c_s}{c_r}\geq \mbar(\mbar+1)(h_{\max}-1)$. 
\end{enumerate} 
\end{proof}

\begin{corollary}
Consider the subclass of MSSN-MC-HC problems where there exists a feasible solution that uses exactly one sink. Further, suppose there exists $\alpha\in (0,1]$ such that $c_s/c_r$ scales as $\frac{c_s}{c_r}\geq \lceil\alpha m\rceil (\lceil\alpha m\rceil + 1)(h_{\max}-1)$. Then, the worst case approximation factor of the SmartSelect algorithm is upper bounded by $(1+\frac{1}{\alpha})$, i.e., for this subclass of problems, SmartSelect provides an $O(1)$ approximation guarantee.
\end{corollary}

\begin{proof}
Putting $\mbar = \lceil\alpha m\rceil$, and proceeding as in the proof of Part 2 of Theorem~\ref{thm:smartselect-approx_special}, the approximation ratio can be upper bounded by

\begin{align}
\left\lceil\frac{m}{\mbar + 1}\right\rceil+\frac{m}{\mbar(\mbar + 1)}&\leq 1+\frac{m}{\mbar + 1}+\frac{m}{\mbar(\mbar + 1)}\nonumber\\
&= 1 + \frac{m}{\mbar}\nonumber\\
&= 1+ \frac{m}{\lceil\alpha m\rceil}\leq 1+\frac{m}{\alpha m}=1+\frac{1}{\alpha}\nonumber
\end{align}
\end{proof}

\begin{theorem}
\label{thm:smartselect-approx_general}
Suppose $\frac{c_s}{c_r}\geq \mbar(\mbar+1)(h_{\max}-1)$, for some $\mbar\in\mathbb{N}$, $m > \mbar \geq 1$. Then, for the general class of MSSN-MC-HC problems, the worst case approximation guarantee of the SmartSelect algorithm is upper bounded by $\max\{\epsilon + \frac{m}{\mbar},\frac{m}{2}\left(1+\frac{1}{\mbar(\mbar + 1)}\right)\}$, where $\epsilon\in[0,1)$ is as defined in Theorem~\ref{thm:smartselect-approx_special}.
\end{theorem}

\begin{proof}
From Theorem~\ref{thm:smartselect-approx_special}, we know that for the subclass of problems where there exists a feasible solution with exactly one sink, the worst case approximation guarantee is upper bounded by $\epsilon + \frac{m}{\mbar}$.

Now consider the remaining class of problems, i.e., the class of problems where any feasible solution uses at least two sinks. Then the optimal solution has cost at least $2c_s$, and hence the worst case approximation guarantee for this class of problems is trivially upper bounded by

\begin{align}
\frac{mc_s\:+\:m(h_{\max}-1)c_r}{2c_s}&= \frac{m}{2}\:+\:m(h_{\max}-1)\frac{c_r}{2c_s}\nonumber\\
&\leq \frac{m}{2}\:+\:\frac{m}{2}\frac{1}{\mbar(\mbar + 1)}\nonumber\\
&= \frac{m}{2}\left(1+\frac{1}{\mbar(\mbar + 1)}\right)\nonumber
\end{align}

The claim follows by combining the upper bounds for the two classes. 
\end{proof}

\subsection{A Destroy and Repair Heuristic to Improve upon SmartSelect}
\label{subsec:d-r}
We propose below a polynomial time heuristic to improve upon the solution provided by the SmartSelect algorithm. The heuristic works by iteratively detroying part of the current solution, and rebuilding the solution by exploring other parts of the search space. Such ideas have been used before to solve hard combinatorial problems; see, for example, \cite{Costa}. We list the detailed steps below.

\gap
\noindent
\textbf{Destroy and Repair Heuristic}
\begin{enumerate}
 \item Let $N^{(0)}$ be the outcome of the SmartSelect algorithm. $N^{(0)}$ is the restriction of the initial network graph $G$ to the sources, selected sinks and selected relays. Set $k = 0$. Also set $N_{best}= N^{(0)}$, and $solution\_update = false$. Let $K$ be the maximum number of iterations allowed.
 \item For each sink $b_j$ in $N^{(k)}$, do the following:
 \begin{itemize}
  \item Pretend to prune the sink $b_j$.
  \item Run the SmartSelect algorithm using \emph{only the remaining sinks in $N^{(k)}$}, and \emph{all} potential relays to obtain a solution $N_1$.
  \begin{itemize}
  \item If $N_1$ is feasible, and cost of $N_1$ is better than cost of $N_{best}$, set $N_{best}=N_1$, and $solution\_update = true$. Else go to the next step.
  \end{itemize}
  \item Pretend to prune the sink $b_j$, and Run the SmartSelect algorithm using \emph{all} the remaining sinks and relays in $G$ to obtain a solution $N_2$.
  \begin{itemize}
   \item If $N_2$ is feasible, and cost of $N_2$ is better than cost of $N_{best}$, set $N_{best}=N_2$, and $solution\_update = true$.
  \end{itemize}
 \end{itemize}
\item After all the sinks in $N^{(k)}$ have been tried (for pruning), if $solution\_update = true$, set $k\leftarrow k+1$, $N^{(k)}=N_{best}$, and go to Step 2.
\item \textbf{Stop} when no further solution update is possible, or the maximum number of iterations have been exceeded. 
\end{enumerate}

\remark Since each iteration of the Destroy and Repair heuristic uses the SmartSelect algorithm which is polynomial time, and since the number of iterations is upper bounded by a constant $K$, it follows that the heuristic is polynomial time. 

\section{Node Cut based ILP Formulation for the MSSN-MC-HC Problem}
\label{sec:ilp-oneconnect}
We shall formulate the MSSN-MC-HC problem as an ILP, using certain node cut inequalities (the approach almost mimicks the one presented in \cite{fullpaper}). Such a formulation will be useful when the number of potential sink and relay locations is considerably large so that a complete enumeration of all possible solutions to obtain the optimal solution (for comparison against the solution provided by the SmartSelect and Destroy-Repair algorithm) is impractical; in such cases, we can solve the LP relaxation of the ILP to obtain a lower bound on the optimal solution for comparison with the SmartSelect and Destroy-Repair outcome. 

\noindent
We start with a couple of definitions.

\begin{definition}
Given a source and a sink in a graph, a \emph{\textbf{node cut} for that source-sink pair} is defined as a set of nodes whose deletion disconnects the source from the sink. 
\end{definition}

\begin{definition}
A \emph{\textbf{minimal node cut} for a source-sink pair} is a node cut which does not contain any other node cut as its subset. 
\end{definition}

Let us make the following construction.

\begin{enumerate}
\item Augment the graph $G$ as follows:
\begin{itemize}
\item Introduce a virtual sink, $b_0$. Let $\tilde{V}=V\cup \{b_0\}$. Let $Q\leftarrow Q\cup \{b_0\}$.
\item Introduce a set of new edges $E^{\prime}=\{(b_0,b_j),1\leq j\leq |B|\}$. Let $\tilde{E}=E\cup E^{\prime}$.
\item Denote the augmented graph by $\tilde{G}=(\tilde{V},\tilde{E})$.
\end{itemize}
\item Using the terminology of RST-MR-HC problem \cite{fullpaper}, define $\tilde{R}= R\cup B$ as the set of potential relay locations, where now, the relays have \emph{non-uniform} costs; each relay in $R$ has cost $c_r$, and each relay in $B$ has cost $c_s$. 
\end{enumerate}

The original sink and relay placement problem now reduces to the following: obtain a tree spanning $Q$, rooted at $b_0$, using a minimum cost subset of $\tilde{R}$, such that the path from each source to $b_0$ has $\leq h_{\max}+1$ hops. 

Consider the graph $\tilde{G} = (Q\cup \tilde{R}, \tilde{E})$. We use the shorthand $0$ to denote $b_0$. We define, $\forall k\in Q\backslash\{0\}$, $\forall j\in \tilde{V}\backslash\{k,0\}$,
\begin{equation*}
y_{j,k}=\left\{
\begin{array}{rl}
1 & \text{if node $j$ is selected to connect source $k$ to the virtual sink}\\
0 & \text{otherwise}
\end{array}\right.
\end{equation*}

Let $\mathcal{P}_k, k\in Q\backslash\{0\}$, denote the set of paths from source $k$ to the sink in the graph $\Gt$. A path $p_k\in \mathcal{P}_k$ from source $k$ to sink is said to be \emph{selected} if $y_{j,k}=1\quad\forall j\in p_k$. A source $k$ is said to be connected to the sink if at least one of the paths in $\mathcal{P}_k$ is selected.
 
\begin{theorem}
The following condition is both \emph{necessary and sufficient} for connectivity of all the sources to the virtual sink:

\begin{equation}
\sum_{j\in \gamma}y_{j,k}\geq 1\quad \forall \gamma\in\Gamma^k;\forall k\in Q\backslash\{0\}
\label{eqn:node-cut-ineq}
\end{equation}
where, $\Gamma^k$ is the set of minimal node cuts for a source node $k$.
\end{theorem}

\begin{proof}
The proof is exactly same as given in \cite{fullpaper}.
\end{proof}

Denote by $c_j$, the cost of node $j\in\Rt$. We now formulate the ILP as follows:

\begin{align}
\min \quad\sum_{j\in \Rt}c_jy_j\label{obj:ilp}\\
\text{Subject to:}\sum_{j\in \gamma}y_{j,k}&\geq 1\quad \forall \gamma\in\Gamma^k;\forall k\in Q\backslash\{0\}\label{constr:conn}\\
y_j &\geq y_{j,k}\quad \forall j\in \Rt;\forall k\in Q\backslash\{0\}\label{constr:nodeselect}\\
\sum_{j\in \Vt\backslash\{k,0\}}y_{j,k}&\leq h_{\max}\quad\forall k\in Q\backslash\{0\}\label{constr:hop}\\
y_{j,k}&\in \{0,1\}\quad \forall k\in Q\backslash\{0\};\forall j\in \Vt\backslash\{k,0\}\label{constr:int1}\\
y_j &\in \{0,1\}\quad \forall j\in \Rt\label{constr:int2} 
\end{align}

Constraint \eqref{constr:conn} in the above formulation ensures connectivity from each source to the virtual sink; constraint \eqref{constr:nodeselect} simply says that a node in $\Rt$ gets selected if it is selected for the path of at least one source; constraint \eqref{constr:hop} ensures that a selected path from a source to the virtual sink has no more than $h_{\max}+1$ hops; constraints \eqref{constr:int1} and \eqref{constr:int2} are the integer constraints on the node selection variables. The objective function \eqref{obj:ilp} simply minimizes the total cost of the nodes selected.

We shall now show that the optimum value of the objective function for the ILP is indeed the same as the optimum solution (i.e., the minimum cost of selected sinks and relays) to the original MSSN-MC-HC problem.

To do that, we introduce the following notations:

\begin{description}
\item $\F = \{\underline{y}=\{\{y_{j,k}\}_{j\in \Vt\backslash\{k,0\},k\in Q\backslash\{0\}},\{y_j\}_{j\in \Rt}\}: \underline{y} \text{ satisfies constraints \eqref{constr:conn}-\eqref{constr:int2}}\}$: set of all feasible solutions to the ILP

\item $\P_k^{'} = \{p_k: \text{$p_k$ consists of $\leq h_{\max}+1$ hops from source $k$ to the virtual sink}\}\subseteq \P_k$: set of all hop count feasible paths from source $k$ to the virtual sink

\item $\U_0 = \{\underline{g}\triangleq \{p_k\}_{k=1}^{|Q|-1}: p_k\in \P_k^{'}\}$: all possible combinations of hop count feasible paths from the sources to the virtual sink
\end{description}

Define a set $\F_0$ in a \emph{one-to-one correspondence} to the set $\U_0$ as follows:

For each $\underline{g}=\{p_k\}_{k=1}^{|Q|-1}\in \U_0$, define $\underline{x}(\underline{g})=\{\{x_{j,k}\}_{j\in \Vt\backslash\{k,0\},k\in Q\backslash\{0\}},\{x_j\}_{j\in \Rt}\}\in\F_0$ such that

\begin{equation*}
x_{j,k}=\left\{
\begin{array}{rl}
1 & \text{if $j\in p_k$}\\
0 & \text{otherwise}
\end{array}\right.
\end{equation*}
\begin{equation*}
x_{j}=\left\{
\begin{array}{rl}
1 & \text{if $x_{j,k}=1$ for some $k\in Q\backslash\{0\}$}\\
0 & \text{otherwise}
\end{array}\right.
\end{equation*}

\begin{lemma}
$\F_0\subseteq \F$
\end{lemma}
\begin{proof}
Verify that any $\underline{x}\in \F_0$ satisfies constraints \eqref{constr:conn}-\eqref{constr:int2}.
\end{proof}

\begin{corollary}
\label{cor:ilpcor1}
$\min_{\underline{y}\in\F}\sum_{j\in \Rt}c_jy_j\leq \min_{\underline{x}\in\F_0}\sum_{j\in \Rt}c_jx_j$
\end{corollary}

Observe that in Corollary~\ref{cor:ilpcor1}, the \emph{L.H.S is the optimum objective function value of the ILP}, whereas the \emph{R.H.S is the optimum solution for the MSSN-MC-HC problem}. Thus, we have proved that the optimum solution to the ILP is a lower bound to the optimum solution to MSSN-MC-HC problem.

\begin{lemma}
For each $\underline{y}\in\F$, $\exists \:\underline{x}\in\F_0$ such that 
\begin{enumerate}
\item $x_{j,k}\leq y_{j,k}\:\forall j,\forall k$, and hence
\item $\sum_{j\in \Rt}c_jx_j\leq \sum_{j\in \Rt}c_jy_j$
\end{enumerate}
\end{lemma}

\begin{proof}
Given $\underline{y}\in \F$, we can construct paths $p_k\in\P_k^{'},\:k\in Q\backslash\{0\}$ such that $\underline{g}=\{p_k\}_{k=1}^{|Q|-1}\in \U_0$. In doing this, we require constraints \eqref{constr:conn} and \eqref{constr:hop} in the definition of $\F$. Now obtain $\underline{x}\in \F_0$ for this $\underline{g}=\{p_k\}_{k=1}^{|Q|-1}\in \U_0$. Observe that $x_{j,k}\leq y_{j,k}\:\forall j,\forall k$.

Also, since the variables are binary, this implies that $\max_{k\in Q\backslash\{0\}}x_{j,k}\leq \max_{k\in Q\backslash\{0\}}y_{j,k}$ for all $j\in \Rt$, i.e., $x_j\leq y_j\:\forall j\in \Rt$. For otherwise, suppose $\max_{k\in Q\backslash\{0\}}x_{j,k} > \max_{k\in Q\backslash\{0\}}y_{j,k}$ for some $j\in \Rt$. Then that would imply, $\max_{k\in Q\backslash\{0\}}x_{j,k}=1$ and $\max_{k\in Q\backslash\{0\}}y_{j,k}=0$, i.e., for that $j\in \Rt$, $\exists \:k\in Q\backslash\{0\}$ such that $x_{j,k}=1$ and $y_{j,k}=0$. But this contradicts the fact that $x_{j,k}\leq y_{j,k}\:\forall j,\forall k$. Hence the conclusion.

Therefore, it follows that $\sum_{j\in \Rt}c_jx_j\leq \sum_{j\in \Rt}c_jy_j$, since $c_j\geq 0$ for all $j\in\Rt$.  
\end{proof} 

\begin{corollary}
\label{cor:ilpcor2}
\begin{equation*}
\min_{\underline{y}\in\F}\sum_{j\in \Rt}c_jy_j\geq \min_{\underline{x}\in\F_0}\sum_{j\in \Rt}c_jx_j
\end{equation*}
\end{corollary}

\begin{proof}
Suppose $\underline{y}_{opt}=\arg\min_{\underline{y}\in\F}\sum_{j\in \Rt}c_jy_j$. Then, by the above lemma, $\exists \underline{x}^{'}\in\F_0$ such that $\sum_{j\in \Rt}c_jx^{'}_j\leq \sum_{j\in \Rt}c_jy_{opt,j}$. But clearly, $\min_{\underline{x}\in\F_0}\sum_{j\in \Rt}c_jx_j\leq \sum_{j\in \Rt}c_jx^{'}_j$. Hence the proof.
\end{proof}

\begin{theorem}
\label{thm:ilp-main}
\begin{equation*}
\min_{\underline{y}\in\F}\sum_{j\in \Rt}c_jy_j\:=\:\min_{\underline{x}\in\F_0}\sum_{j\in \Rt}c_jx_j
\end{equation*}
\end{theorem}

\begin{proof}
The proof follows by combining Corollaries~\ref{cor:ilpcor1} and \ref{cor:ilpcor2}.
\end{proof}

Theorem~\ref{thm:ilp-main} states that the optimum value of the objective function for the ILP is indeed the same as the optimum solution to the original MSSN-MC-HC problem.

To solve the LP relaxation of this ILP to obtain a lower bound on the optimal solution, we use the algorithm presented in \cite{nigam} (with the Master problem being the ILP represented by Equations~\eqref{obj:ilp}-\eqref{constr:int2}), which uses as a sub-program (to find the node cut constraints iteratively), an algorithm presented by Garg et al.\ \cite{garg} in the context of node weighted multiway cuts.

\section{Numerical Results} 
\label{sec:results}

Since the MSSN-MC-HC problem is NP-Hard, computing the optimal solution would require, in general, an exhaustive search over all possible combinations of potential sinks and relays, which is clearly impractical. We, therefore, obtain a lower bound on the optimum cost of a problem instance by solving the LP relaxation of the ILP formulation for the MSSN-MC-HC problem. 

We compare the performance of our algorithms against this LP-based lower bound on the optimum cost, as well as the exponential search heuristic (GRASP-MSRP) proposed by Sitanyah et al.\ \cite{sitanayah} in three different experimental settings (obtained by varying one or more of the following: number of sources, potential relays and potential sinks, deployment area, communication range, hop constraint, and strategy for choosing the source node locations and potential node locations). In all the experiments, we chose $c_s = 10$, and $c_r = 1$; note that these choices satisfy the hypothesis on $c_s/c_r$ made in Theorems~\ref{thm:smartselect-approx_special} and \ref{thm:smartselect-approx_general} with $\mbar = 1$. The details of the experimental settings are provided in Table~\ref{tbl:expt-settings}, where, by ``random placement in grid'', we mean the following: the area is partitioned into square cells of side 10$m$. Consider the lattice created by the corner points of the cells. For each of the instances under this setting, the source nodes, the potential relay locations, and the potential sink locations are selected at random from these lattice points. Similarly, ``random placement without grid'' means that the source locations, potential relay and sink locations were picked according to a continuous uniform distribution over the area (without partitioning the area into grids).

\begin{table*}[t]
  \centering
\caption{Details of the Experimental Settings}
\label{tbl:expt-settings}
\scriptsize
  \begin{tabular}{|c|c|c|c|c|c|c|c|c|}\hline
    Setup & Instances & Sources & potential & potential & Locations & Area & $r_{\max}$ & $h_{\max}$\\
    & generated & & relays & sinks & & (in $m^2$) & (in meters) & \\
    \hline
   1 & 60 & 20 & 30 & 10 & random placement & 100$\times$ 100 & 20 & 5\\
     &    &    &    &    &  without grid    &               &    &   \\ 
  \hline
   2 & 60 & 40 & 50 & 15 & random placement & 140$\times$ 140 & 20 & 5\\
     &    &    &    &    &  without grid    &                 &    &   \\
  \hline
   3 & 60& 30 & 50 & 15 & random placement & 140$\times$ 140 & 30 & 5\\
     &    &    &    &    &  in grid         &                 &    &   \\
\hline
\end{tabular}
\normalsize
\end{table*}    

For ease of exposition, from now on, we use the abbreviations SS, DR, and GM respectively to denote the SmartSelect algorithm, the Destroy and Repair heuristic, and the GRASP-MSRP heuristic proposed in \cite{sitanayah}. For each instance in each setting, we ran all the three algorithms on that instance, and also computed the LP-based lower bound on the optimum solution for that instance. The maximum number of iterations in the DR heuristic was chosen to be 25. The experiments were run using MATLAB R2011b on a Linux based desktop with 8 GB RAM. Table~\ref{tbl:lp-compare} summarizes the performance of the algorithms as compared to the LP-based lower bound. For each of the feasible instances in each setting, we computed the empirical approximation ratio of each algorithm (with respect to the LP-based lower bound) as $\text{Approx. ratio} = \frac{\text{Cost of the algorithm outcome}}{\text{Cost of LP solution}}$. 
The maximum of these over the instances in that setting was taken as the empirical worst case approximation ratio for that setting. 

We also computed the empirical average case approximation ratio of the algorithms as follows: let $C^{(algo)}_{avg}$ be the average cost of the algorithm outcome over the feasible instances in a setting, and let $\overline{C}_{lp}$ be the average cost of the LP solution over those feasible instances. Then, the empirical average case approximation ratio, $\overline{\alpha}_{algo}$ was computed as $\overline{\alpha}_{algo}=\frac{C^{(algo)}_{avg}}{\overline{C}_{lp}}$. The theoretical upper bound on the worst case approximation ratio of SS algorithm was computed using Theorem~\ref{thm:smartselect-approx_general}.

\begin{table*}[t]
  \centering
\caption{Performance Comparison of the Algorithms against LP-based Lower Bound}
\label{tbl:lp-compare}
\scriptsize
  \begin{tabular}{|c|c|c|c|c|c|c|c|c|}\hline
    Experimental & feasible & \multicolumn{3}{c|}{Empirical} & \multicolumn{3}{c|}{Empirical} & Theoretical\\
    setup & instances & \multicolumn{3}{c|}{average case} & \multicolumn{3}{c|}{worst case} & worst case\\
   & & \multicolumn{3}{c|}{approx. ratio} & \multicolumn{3}{c|}{approx. ratio} & approx. ratio bound\\
   & & SS & DR & GM & SS & DR & GM & of SS\\
    \hline
   1 & 47 & 1.353 & 1.29& 1.255 & 2.88 & 2.88 & 2.72 & 20\\
   2 & 32 & 1.11 & 1.05 & 1.005 & 1.358& 1.205& 1.057& 40\\
   3 & 60& 1.09 & 1.012& 1.007 & 1.83 & 1.25 & 1.08 & 30 \\ 
  \hline
\end{tabular}
\normalsize
\end{table*} 
In Table~\ref{tbl:run-time}, we compare the running times of the algorithms and that of the LP-based lower bound computation.
\begin{table*}[t]
  \centering
\caption{Execution Times of the Algorithms, and the LP-based Lower Bound Computation}
\label{tbl:run-time}
\scriptsize
  \begin{tabular}{|c|c|c|c|c|c|c|c|c|}\hline
    Experimental &  \multicolumn{8}{c|}{Execution time in secs} \\
    setup & \multicolumn{2}{c|}{SS} & \multicolumn{2}{c|}{DR} & \multicolumn{2}{c|}{GM} & \multicolumn{2}{c|}{LP}\\
   & mean & max & mean & max & mean & max & mean & max\\
    \hline
   1 & 1.148 & 2.72 & 3.4565 & 15.0878 & 258.645 & 673.415 & 37.36 & 138.14\\
   2 & 4.5344 & 12.37 & 23.64 & 86.02 & 1885.8 & 6296.4 & 3819.7 & 41785\\
   3 & 5.173 & 11.46 & 12.095 & 34.922& 311.514& 638.15 & 489.039& 5395.3\\ 
  \hline
\end{tabular}
\normalsize
\end{table*}

From Tables~\ref{tbl:lp-compare} and \ref{tbl:run-time}, we make the following observations:
\begin{enumerate}
\item In all the experimental settings considered, the \emph{average} empirical performance of both the SS and DR algorithms in terms of cost are within a factor of about 1.4 of the LP based \emph{lower bound} on the optimum cost. Notice that the actual performance would be even better since we are only comparing against a lower bound on the optimum cost. In the worst case, the algorithms are off from the lower bound by a factor of about 2.9, which is still much better than the theoretically predicted performance bound for the SS algorithm in the corresponding setting. 
\item The performance of the DR algorithm is better than that of the SS algorithm as expected, although at the cost of a slightly higher running time.
\item Both the SS and the DR algorithms achieve \emph{orders of magnitude improvement in running time} compared to the LP (and hence, obviously with respect to the exact ILP). 
\item While the GM algorithm (\cite{sitanayah}) does marginally better than the SS and DR algorithms in terms of cost, the improvement comes at a heavy price in terms of running time. In all the settings considered, both the SS and DR algorithms are order of magnitude faster compared to the GM algorithm.
\end{enumerate}

This extremely fast running time makes SmartSelect and Destroy-Repair, an attractive choice for use in a Field-interactive iterative network design tool such as SmartConnect \cite{smartconnect-paper}. However, one might argue that this improvement in running time comes at the price of a degradation in performance (i.e., cost of the resulting solution). We, therefore, proceed to further quantify the degradation in cost when DR (or SS) algorithm is used instead of the GM algorithm. To this end, we first quantify, from our experimental data, the improvement in cost achieved by the DR heuristic over SS algorithm. Our findings are summarized in Table~\ref{tbl:dr-ss}.

\begin{table*}[t]
  \centering
\caption{Performance Comparison of the DR and SS Algorithms}
\label{tbl:dr-ss}
\scriptsize
  \begin{tabular}{|c|c|c|}\hline
    Experimental & Improvement in average cost & Maximum improvement in cost\\
    setup &  by DR over SS &  by DR over SS\\
          &   (in \%) & (in \%)\\
    \hline
   1 & 4.87 & 37.5\\
   2 & 5.65 & 29.3\\
   3 & 7.68 & 75\\ 
  \hline
\end{tabular}
\normalsize
\end{table*} 

Since we observe from Table~\ref{tbl:dr-ss} that the DR algorithm achieves an improvement in average cost of about 5\% to 8\% (and a maximum improvement of 75\% over all the instances) over the SS algorithm in all the settings considered, and has the same order of running time as the SS algorithm (as observed from Table~\ref{tbl:run-time}), we next compare the performance of the DR algorithm against that of the GM algorithm (\cite{sitanayah}) in the scenarios tested. Our findings are summarized in Table~\ref{tbl:dr-gm}.

\begin{table*}[t]
  \centering
\caption{Performance Comparison of the DR and GM Algorithms}
\label{tbl:dr-gm}
\scriptsize
  \begin{tabular}{|c|c|c|c|}\hline
    Experimental & Degradation in & Maximum degradation & Maximum improvement\\
      setup      &  average cost  &    in cost          &  in cost            \\
                 & of DR w.r.t GM (in \%)&   of DR w.r.t GM (in \%)   &  of DR w.r.t GM (in \%)     \\
    \hline
      1 & 2.76 & 33.33 & 11.54\\
      2 & 4.55 & 20.51 & 5.36\\
      3 & 0.52 & 25    & 7.69\\ 
  \hline
\end{tabular}
\normalsize
\end{table*} 

From Table~\ref{tbl:dr-gm}, we make a couple of observations:

\begin{enumerate}
\item In all the settings considered, the \emph{average} cost of the DR heuristic is within at most 4.6\% of that of the GM algorithm, and in the worst case (over all the tested scenarios in all the settings), the cost of the DR algorithm is off by 33.33\% from that of the GM algorithm. On the other hand, as can be observed from Table~\ref{tbl:run-time}, the improvement in average run time of the DR algorithm compared to that of the GM algorithm is up to a factor of about 80. 
\item A perhaps surprising observation is that in all the three settings considered, there were instances where the DR algorithm in fact did better than the (more complex) GM algorithm even in terms of cost, as indicated by the last column in Table~\ref{tbl:dr-gm}, and the improvement was upto 11.54\%. 
\end{enumerate}

In summary, we conclude that the DR heuristic (with the SS algorithm as its starting point) achieves significant improvement in running time compared to both the GM algorithm (\cite{sitanayah}), and ILP based solutions, while incurring only a small penalty in terms of cost. This extremely fast running time, and insignificant penalty in cost make the Destroy and Repair heuristic (with SmartSelect as starting point) an excellent choice for use in an iterative network design procedure such as SmartConnect \cite{smartconnect-paper}.

\section{Conclusion}
\label{sec:conclude}
In this paper, we have studied the problem of determining an optimal relay and sink node placement strategy such that certain performance objective(s) (in this case, hop constraint, which, under a lone-packet model, ensures data delivery to the BS within a certain maximum delay) is (are) met. We found that the problem is NP-Hard, and is even hard to approximate within a factor of $O(\ln m)$, where $m$ is the number of sources. We have proposed a polynomial time approximation algorithm for the problem. The algorithm is simple, intuitive, and as can be concluded from numerical experiments presented in Section~\ref{sec:results}, gives solutions of very good quality in extremely reasonable computation time. We have also provided worst case bound on the performance of the algorithm. 

Further, we are working on combining our algorithm with that proposed by Sitanayah et al. \cite{sitanayah} to further improve the cost efficiency of our algorithm while retaining the benefits of fast running time. 
  
\bibliographystyle{IEEE}
\bibliography{dit_astec}
\end{document}